\documentclass{daj}

\dajAUTHORdetails{%
  title ={An Improved Lower Bound for Sparse Reconstruction from Subsampled Walsh Matrices}
  author = {Jaros\l aw B\l asiok, Patrick Lopatto, Kyle Luh, Jake Marcinek and Shravas Rao},
  plaintextauthor = {Jaroslaw Blasiok, Patrick Lopatto, Kyle Luh, Jake Marcinek and Shravas Rao},
  runningtitle = {An Improved Lower Bound for Sparse Reconstruction}, 
  runningauthor ={Jaros\l aw B\l asiok, Patrick Lopatto, Kyle Luh, Jake Marcinek and Shravas Rao},
  copyrightauthor = {Jaros\l aw B\l asiok, Patrick Lopatto, Kyle Luh, Jake Marcinek and Shravas Rao},
  keywords = {compressed sensing, sparse recovery, restricted isometry property},
}   %%% END \dajAUTHORdetails

\dajEDITORdetails{%
   year={2023},
   number={3},
   received={25 June 2020},   % received date: example: 7 January 2017
   published={10 May 2023},  % published date
   doi={10.19086/da.74802},       % XXX = number of paper, e.g. da006 for paper#6
}   %%% END \dajEDITORdetails

\usepackage{amsmath}
\usepackage{todonotes}
\usepackage{amsthm}
\usepackage{amssymb,prettyref,appendix,xcolor,graphicx}
\usepackage{comment}
\usepackage{microtype}
\usepackage{bbm}
\theoremstyle{plain}
\newtheorem{thm}{Theorem}[section]
  \newtheorem{lem}[thm]{Lemma}

  \newtheorem{claim}[thm]{Claim}
\theoremstyle{definition}

\usepackage{mathrsfs}
\DeclareFontFamily{U}{mathx}{\hyphenchar\font45}
\DeclareFontShape{U}{mathx}{m}{n}{<-> mathx10}{}
\DeclareSymbolFont{mathx}{U}{mathx}{m}{n}
\DeclareMathAccent{\widebar}{0}{mathx}{"73}

\newcommand{\Oh}{\mathcal{O}}

\newcommand{\R}{\mathbb{R}}

\newcommand{\Z}{\mathbb{Z}}

\newcommand{\F}{\mathbb{F}}

\newcommand{\E}{\mathbb{E}}

\newcommand{\Cov}{\text{Cov}}
\newcommand{\Var}{\operatorname{Var}}

\DeclareMathOperator{\argmin}{arg\,min}

\newcommand{\Gr}[1]{\mathbb{G}_{#1}}
\newcommand{\one}{\mathbbm{1}}

\def\bp{\begin{proof}}
\def\ep{\end{proof}}

\def\P{\mathbb{P}}
\def\E{\mathbb{E}}

\begin{document}

\begin{frontmatter}[classification=text]
%% EDITOR: this will force the keywords to appear right after the Abstract.
%%   If the abstract is too long and would force the keywords off the
%%   front page, please comment out % [classification=text] above
%%   This way the keywords will be floated on the bottom of the first page
%%   even though the Abstract spills over to the next page.

%%% AUTHOR: Title goes here.  This line is optional.  You must use it
%%   if title has footnote attached or requires nontrivial typesetting,
%%   e.g., inclusion of linebreaks to force nice layout.
\title{An Improved Lower Bound for Sparse Reconstruction from Subsampled Walsh Matrices} %% please capitalize all significant words

%%% AUTHOR:
%%% List all authors. If you wish, place grant acknowledgements in \thanks.
%%% In brackets include a short tag for each author.
\author[pgom]{Jaros\l aw B\l asiok\thanks{Supported by Siebel Foundation.}}
\author[johan]{Patrick Lopatto\thanks{Partially supported by the NSF Graduate Research Fellowship Program under grant DGE-1144152.}}
\author[laci]{Kyle Luh\thanks{Partially supported by NSF Postdoctoral Fellowship DMS-1702533.}}
\author[andy]{Jake Marcinek}
\author[andy]{Shravas Rao\thanks{Supported by the Simons Collaboration on Algorithms and Geometry.}}

%%% AUTHOR: Abstract goes here
\begin{abstract}
		We give a short argument that yields a new lower bound on the number of uniformly and independently subsampled rows from a bounded, orthonormal matrix necessary to form a matrix with the restricted isometry property.  We show that a matrix formed by uniformly and independently subsampling rows of an $N \times N$ Walsh matrix contains a $K$-sparse vector in the kernel, unless the number of subsampled rows is $\Omega(K \log K \log (N/K))$ --- our lower bound applies whenever $\min(K, N/K) > \log^C N$. Containing a sparse vector in the kernel precludes not only the restricted isometry property, but more generally the application of those matrices for uniform sparse recovery.
\end{abstract}
\end{frontmatter}

%%% AUTHOR: body of paper starts here
\section{Introduction}
In their seminal work on sparse recovery \cite{candes2006near}, Cand\`es and Tao were led to the notion of the \emph{restricted isometry property} (RIP).  A $q \times N$ matrix $M$ has the restricted isometry property of order $K$ with constant $\delta >0$ if for all $K$-sparse vectors $x \in \mathbb{C}^N$ (i.e. vectors with at most $K$ nonzero entries) we have
$$
(1-\delta) \|x\|_2^2 \leq \|M x\|_2^2 \leq (1+\delta) \|x\|_2^2.
$$ 

The significance of this property is that it guarantees that one can recover an approximately $K$-sparse vector $x^*$ from $Mx^*$ via a convex program \cite{candes2006near}. Specifically, they showed that if a matrix $M$ satisfies $(2K, \sqrt{2}-1)$-RIP, then the minimizer 
\begin{equation*}
		\tilde{x} := \argmin_{x : Mx = Mx^{*}} \|x\|_1,
\end{equation*}
satisfies 
\begin{equation*}
		\|\tilde{x} - x^*\|_2 \leq \frac{1}{\sqrt{k}} \|x^{*} - x^{*}_K\|_1,
\end{equation*}
where $x^{*}_K$ is the best $K$-sparse approximation of $x^*$ --- in particular when $x^*$ is exactly $K$-sparse, it can be efficiently recovered from $Mx^*$ without any error.

In applications, $q$ is the number of measurements needed to recover a sparse signal. Therefore, it is of interest to understand the minimal number of rows needed in a matrix with the RIP property.  

It is known that for a properly normalized matrix with independent gaussian entries, $q = \Oh(K \log(N/K))$ suffices to generate a RIP matrix with high probability (e.g. \cite{foucart2013compressive}).  Yet, it is often beneficial to have more structure in the matrix $M$ %, and matrices with orthonormal columns and bounded entries (with the proper normalization) have received the most attention in the applied mathematics community 
\cite{rauhut2010compressive}. For example, if the matrix $M$ is a submatrix of the discrete 
Fourier transform matrix, then the fast Fourier transform algorithm allows fast matrix--vector multiplication, speeding up the run time of the recovery algorithm \cite[Chapter 12]{foucart2013compressive}.  Additionally, generating a random submatrix requires fewer random bits and less storage space.   %The primary example of such a matrix is the discrete Fourier transform (DFT).  

The first bound on the number of uniformly and independently subsampled rows from a Fourier matrix necessary for recovery appeared in the groundbreaking work \cite{candes2006near}.  They showed that if one randomly subsamples rows so that the expected number of rows is $\Oh(K \cdot \log^6 N)$, then concatenating these rows forms a RIP matrix with high probability, after appropriate normalization.  Rudelson and Verhsynin later improved this bound to  $\Oh(K\cdot \log^2 K \cdot \log(K \log N) \cdot \log N)$ via a gaussian process argument involving chaining techniques \cite{rudelson2008sparse}.  Their proof was then streamlined and their probability bounds strengthened \cite{dirksen2015tail, rauhut2010compressive}. Cheraghchi, Guruswami, and Velingker then proved a bound of $\Oh(K \cdot \log^3 K \cdot \log N)$ \cite{cheraghchi2013rip}, and Bourgain established the bound $\Oh(K \cdot \log K \cdot \log^2 N)$ \cite{bourain2014rip}. The sharpest result in this direction is due to Haviv and Regev, who showed the upper bound $\Oh(K\cdot \log^2 K \cdot \log N)$ through a delicate application of the probabilistic method \cite{haviv2017restricted}.  It is widely conjectured that for the discrete Fourier transform $q = \Oh(K \log N)$ suffices \cite{rudelson2008sparse}.

It turns out that all proofs in this line of work, including the strongest known upper bound \cite{haviv2017restricted}, apply in a more general setting where random matrix $M$ is constructed by uniformly and independently subsampling rows of any \emph{bounded orthonormal matrix} --- that is an orthonormal matrix with all entries bounded in magnitude by $\frac{B}{\sqrt{N}}$ for some constant $B$. The matrix of the Discrete Fourier Transform satisfies this property with $B=1$.

This paper addresses the problem of determining a \emph{necessary} number of samples for reconstruction. Our contribution is that --- surprisingly --- for general bounded orthonormal matrices, and for a certain range of $K$, $\Omega( K \log^2 N)$ samples are needed when sampling uniformly and independently. In particular, only a gap of $\log K$ remains between our bound and the best known upper bound. We improve the previous best lower bound $\Omega(K\cdot \log N)$ due to Bandeira, Lewis, and Mixon \cite{bandeira2017discrete} which applied to the DFT matrix. Those in turn improve upon more general lower bounds $\Omega ( K \cdot \log(N/K))$ on the number of rows for \emph{any} matrix that satisfies the RIP property \cite{ba2010lower, garnaev1984widths, kashin1977diameters, nelson2013sparsity}.  

In the proof we consider an example of a bounded orthonormal matrix, the Walsh matrix (i.e. the matrix of the Fourier transform on the additive group $\Z_2^n$), and we show that for this specific matrix at least $\Omega( K \log K \log N/K)$ samples is required. More concretely, by a second moment argument, we demonstrate that for fewer than $\Oh(K \log K \log N/K)$ subsampled rows, with high probability there exists a $K$-sparse vector in the kernel --- ruling out both the RIP property, and in general any hope for sparse recovery algorithm with those matrices. The same proof can be applied more generally to show that for any prime $r$ one needs to subsample at least $\Omega(K \log K \log(N/K) / \log(r))$ rows of a matrix corresponding to Fourier transform on the additive group $\Z_r^n$ --- for the sake of simplicity of the argument we do not elaborate on this.
  
\section{Preliminaries}
Throughout this note, we use $\log$ to denote the base $2$ logarithm.
For an integer $n\ge 1$, we set $N = 2^n$ and fix a bijection between $[N]$ and $\Z_2^n$; this identification remains in force for the rest of the paper.

We say a function $\chi: \Z_2^n \rightarrow \{\pm 1\}$ is a \emph{character} if it is a group homomorphism.  To an element $a \in \Z_2^n$, we associate the character 
$$
\chi_a(x) = (-1)^{\langle a, x \rangle}
$$
for all $x \in \Z_2^n$.  
The Fourier transform of a function $f : \Z_2^n \rightarrow \mathbb{C}$ is defined to be 
$$
\hat{f}(a) =  \frac{1}{\sqrt{N}} \sum_{x \in \Z_2^n} f(x) \chi_a(x)
$$ 
for all $a \in \Z_2^n$.  
Let $H$ be the $N \times N$ matrix representing the Fourier transform on the group $\Z_2^n$.  In other words,
$$
H_{ij} = \frac{1}{\sqrt{N}} (-1)^{\sum_{k=1}^n i_k j_k}.
$$
When normalized to have $\pm 1$ entries, the matrix $H$ is also known as a Walsh matrix.
We refer the reader to \cite{terras1999fourier} for a thorough discussion of
 Fourier analysis on finite groups.

The Grassmannian $\Gr{n,d} = \Gr{n,d}(\Z_2)$ is defined as the collection of vector subspaces of $\Z_2^n$ of dimension $d$. 
Our proof uses the following well-known result about the Fourier transform.
\begin{lem} \label{lem:ortho}
For a subspace $V \in \Gr{n,d}$, we let $\one_V \in \R^N$ be the vector corresponding to the indicator function for $V$ with the normalization $\| \one_V \|_2 =1$. Then
$$
H \one_V = \one_{V^\perp}.
$$
where $V^\perp$ is the orthogonal complement of $V$.
\end{lem}
In this way, $H$ implements a bijection between $\Gr{n,d}$ and $\Gr{n,n-d}$. 
We also make use of the following bounds on the size of $\Gr{n,d}$.
\begin{lem}\label{countgrass}
The size of $\Gr{n,d}$ is bounded by
\begin{equation}
2^{d(n-d)} < |\Gr{n,d}| < 2^{d(n-d+1)}.
\end{equation}
\end{lem}
\begin{proof}
A standard counting argument gives the explicit formula
\begin{equation}
|\Gr{n,d}| = \prod_{k = 0}^{d-1} \frac{2^n - 2^k}{2^d - 2^k}.
\end{equation}
Using the inequalities
\begin{equation}
2^{n-d} < \frac{2^n - 2^k}{2^d - 2^k} < 2^{n - d + 1}
\end{equation}
on each factor individually gives the result.
\end{proof}

We also make use of the following trivial counting lemma.
\begin{lem} \label{lem:intersection}
For $U, V \in \Gr{n,k}$, 
$$
\max(n-2k, 0) \leq \dim(U^\perp \cap V^\perp) \leq n-k.
$$
\end{lem}

\section{Main Result}
For a subset $Q \subset [N]$, we let $H_Q$ denote the matrix generated from the rows of $H$ indexed by $Q$.  Let $\delta_1, \dots, \delta_N$ be a set of independent Bernoulli random variables which take the value $1$ with probability $\hat{p}$.  Random variables $\delta_i$ will indicate which rows to include in our measurement matrix, $H_Q$, meaning
$$
Q = \{j \in [N]: \delta_j = 1\}.
$$  
Note that $Q$ has average cardinality $N \hat{p}$ and standard concentration arguments can be used to obtain sharp bounds on its size.
We say that a vector $v\in \R^N$ is $K$-sparse if it has at most $K$ nonzero entries. The following theorem is our main technical result. 

\begin{thm}\label{thm:kernel}
		For $\min(k, n-k) \geq 12 \log n$, where $N = 2^n$ and $K = 2^k$, there exists a positive constant $c > 0$ such that for $\hat{p} \leq  \frac{c K}{N} \log K \log(N/K)$, there exists a $K$-sparse vector in the kernel of $H_Q$ with probability $1 - o(1)$.\footnote{$o(1)$ indicates a quantity that tends to zero as $N \rightarrow \infty$.  All asymptotic notation is applied under the assumption that $N \rightarrow \infty$.}
\end{thm}
\bp
We will define $p := - \ln (1 - \hat{p})$ for future convenience, and note that $\hat{p} \leq p \leq 2\hat{p}$, for small enough $\hat{p}$. 

  We restrict our attention to the $K$-sparse vectors that correspond to $\one_{V}$ for $V \in \Gr{n,  k}$, the indicator functions of subspaces of dimension $k$. For such $V$, set $X_V$ to be the indicator function for the event that $Q \cap V^\perp =\emptyset$. Define
 \begin{equation}
 X = \sum_{V \in \Gr{n, k}} X_V.
 \end{equation}
 Observe that by Lemma \ref{lem:ortho}, if $X$ is non-zero then there exists a $K$-sparse vector in the kernel of $H_Q$.
 We proceed by the second moment method to show that $X$ is nonzero with high probability. By the second moment method (e.g. \cite{alonspencer}),
 \begin{equation}
 \P(X = 0) \le \frac{\Var X }{(\E X)^2}.
 \label{eq:prob-bound}
 \end{equation}
We can easily obtain an expression for the first moment:
\begin{align*} 
\E X  &=  | \Gr{n, k} | \cdot \E X_V  \\
&=   | \Gr{n, k}| \left( 1 - \hat{p}  \right)^{|V^\perp|} \\
&= | \Gr{n, k}| \exp(-p\frac{N}{K}) \\
& \geq \exp( (\ln 2 - 2 c) k(n-k)).
\end{align*}

The second moment requires a more delicate calculation.  We partition the sum into pairs of orthogonal complements with the same dimension of intersection.  By Lemma \ref{lem:intersection}, and letting $d_0$ denote $ \max(n-2m, 0)$, we have
\begin{align}
		\frac{\Var X}{(\E X)^2}  &= \frac{\sum_{U, V \in \Gr{n,k}} \Cov(X_U, X_V)}{|\Gr{n,k}|^2 \left(\E X_U\right)^2} \nonumber \\
		& = \frac{\sum_{d=d_0}^{n-k} \sum_{U,V: \dim(U^\perp \cap V^{\perp}) = d} \Cov( X_U X_V)  }{|\Gr{n,k}|^2 \left(\E X_U\right)^2}. \label{eq:sum1}
\end{align}

We can explicitly compute each term in the sum above as follows.
\begin{claim}
		For $U, V \in \Gr{n,k}$ such that $\dim(U^{\perp} \cap V^{\perp}) = d$, we have
		\begin{equation*}
				\frac{\Cov(X_U, X_V)}{\left(\E X_U\right)^2} = \exp(p2^d) - 1.
		\end{equation*}
\end{claim}
\begin{proof}
		Observe that
		\begin{align*}
				\E X_U X_V & = \P(U^{\perp} \cap Q = \emptyset \land V^{\perp} \cap Q = \emptyset) \\
				& = \exp(-p |U^{\perp} \cup V^{\perp}|) \\
				& = \exp(-2 p |U^{\perp}| + p |U^{\perp} \cap V^{\perp}|) \\
				& = \left(\E X_U\right)^2 \exp(p 2^d).
		\end{align*}
\end{proof}
We plug this expression back to the sum \eqref{eq:sum1}, in order to arrive at
\begin{equation*}
		\frac{\Var X}{(\E X)^2} = \sum_{d = d_0}^{n-k} \sum_{U, V: \dim(U^{\perp} \cap V^{\perp}) = d} \frac{1}{|\Gr{n,k}|^2} \left(\exp(p 2^d) - 1\right).
\end{equation*}

Let us use $T(n, k, d)$ to denote number of pairs $U, V \in \Gr{n, k}$ such that $\dim (U^{\perp} \cap V^{\perp}) = d$. With this notation, the entire sum simplifies to
\begin{equation*}
		\sum_{d = d_0}^{n-k} \frac{T(n,k, d)}{|\Gr{n,k}|^2} \left(\exp(p 2^d) - 1\right).
\end{equation*}

We will split this sum into two parts and bound them separately
\begin{align*}
		%\sum_{d=d_0}^{n-k} \frac{T(n,k,d)}{|\Gr{n,k}|^2} \left(\exp(p 2^d) - 1\right) 
		%\frac{\Var(X)}{(\E X)^2}
			& \sum_{d=d_0}^{n - k - 3 \log n} \frac{T(n,k,d)}{|\Gr{n,k}|^2} \left(\exp(p 2^d) - 1\right) \\
			& \qquad + \sum_{d = n - k - 3 \log n}^{n-k} \frac{T(n,k,d)}{|\Gr{n,k}|^2} \left(\exp(p2^d) - 1\right) \\
			& =: (I) + (II).
\end{align*}

The first part of the summation is easy to control: for $d < n - k - 3 \log n$ we have $p 2^d \leq \frac{2 c}{n}$, which implies $\exp(p 2^d) - 1 \leq \frac{4c}{n}$, and
\begin{align}
		(I) & \leq \sum_{d = d_0}^{n-k - 3 \log n} \frac{T(n,k,d)}{|\Gr{n,k}|^2} \frac{4c}{n} \nonumber\\ 
		& \leq \frac{4c}{n} \sum_d \frac{T(n,k,d)}{|\Gr{n,k}|^2} \nonumber \\ 
		& = \frac{4c}{n} = o(1). \label{eq:I-bound} 
\end{align}

We can now turn our attention to bounding $(II)$.
\begin{claim}
		\label{claim:T-bound}
		For $d \geq n - k - 3 \log n$, we have
		\begin{equation*}
				\frac{T(n,k,d)}{|\Gr{n,k}|^2} \leq \exp\left(-\frac{\ln(2)}{2} k (n-k)\right).
		\end{equation*}
\end{claim}
\begin{proof}
		First, we have the bound $T(n,k,d) \leq |\Gr{n,d}| |\Gr{n - d, n - k - d}|^2$. Indeed, to choose two subspaces $U^{\perp}, V^{\perp}$ of dimension $k$ with $\dim(U^{\perp} \cap V^{\perp}) = d$, we can first choose $T = U^{\perp} \cap V^{\perp}$ as a subspace of $\F_2^n$ (there are $|\Gr{n,d}|$ ways of doing this), and then we can consider the quotient space $\F_2^n/T$ and count the number of disjoint subspaces $U/T, V/T \subset \F_2^n / T$ --- the number of such choices is upper bounded by $|\Gr{n-d, n-k-d}|^2$ --- the number of all pairs of subspaces $U/T, V/T \in \F_2^n/T$.

		Applying Lemma~\ref{countgrass} to $|\Gr{n,d}|$ and $|\Gr{n-d, n-k-d}|$, we obtain
		\begin{equation*}
				T(n,k,d) \leq \exp\left( \ln(2)\left[d(n-d + 1) + 2 (n-k-d+1)(k+1)\right]\right).
		\end{equation*}

		The quadratic in the exponent is maximized for $d = \frac{n - 2k - 1}{2}$, hence in the range $d \geq n - k - 3 \log n$, the maximum is attained exactly at $d = n - k - 3 \log n$. This yields
		\begin{align*}
				T(n,k,d) & \leq \exp\left( \ln(2) \left[ (n-k - 3\log n)(k + 3 \log n + 1) + 2(3 \log n + 1)(k + 1)\right]\right) \\
				& \leq \exp\left( \ln(2) \left[ (n-k)(\frac{5}{4} k) + \frac{1}{4} (n - k)k\right]\right) \\
				& \leq \exp\left( \ln(2) \left[\frac{3}{2} (n- k)k\right]\right),
		\end{align*}
		where the second inequality follows from the fact that $\min(k, n-k) \geq 12 \log n$.

		On the other hand, using Lemma~\ref{countgrass} again, we have $\frac{1}{|\Gr{n,k}|^2} \leq 2^{-2 k (n-k)}$ and the statement of the claim follows by combining these two inequalities.
\end{proof}

Now we can introduce a simple upper bound of the sum $(II)$
\begin{align}
	 \sum_{d = n - k - 3 \log n}^{n-k} \frac{T(n,k,d)}{|\Gr{n,k}|^2}\left(\exp(p2^d) - 1\right) & \leq 3 (\log n) 2^{- \frac{1}{2} k(n-k)} \exp\left(p 2^{n-k}\right) \nonumber \\
	 & \leq 3 (\log n) 2^{(\frac{2c}{\ln(2)} - \frac{1}{2}) k (n-k)} \nonumber \\
	 & \leq o(1), \label{eq:II-bound}
\end{align}
where the first inequality follows from Claim~\ref{claim:T-bound}, the second one follows from $p 2^{n-k}  < 2 c k (n-k)$, and the third can one be applied as soon as $\frac{2c}{\ln(2)} < \frac{1}{2}$. The statement of the Theorem now follows by combining \eqref{eq:prob-bound}, \eqref{eq:I-bound} and \eqref{eq:II-bound}.
\ep

We can now state our main result in terms of sparse recovery.
\begin{thm}
Let $N$ and $K$ be as in Theorem \ref{thm:kernel}.
For there to exist a method to recover every $K$-sparse vector from $H_Q$, for any $K$ such that $\min(K, N/K) \geq \log^{12} N$, the expected cardinality of the number of rows of $H_Q$ must be $\Omega(K \log K \log  (N/K) )$.  Further, for any constant $\delta>0$, the expected number of rows of $H_Q$ must be $\Omega(K \log K \log (N/K))$ for $H_Q$ to have the RIP property.
\end{thm}

\begin{proof}
		By Theorem \ref{thm:kernel}, there exists a $2K$-sparse vector $x$ in the kernel of $H_Q$ with high probability if the expected number of rows of $H_Q$ is $o(K \log K \log(N/K))$.  Let us write $x = y - z$ where $y$ and $z$ are both $K$-sparse vectors.  Then $H_Q y = H_Q z$, which proves that $H_Q$ is not injective when restricted to the set of all $K$-sparse vectors.  The statement about the RIP property follows directly from the definition.
\end{proof}

%%% AUTHOR: optional acknowledgments here
%\section*{Acknowledgments} %%  you may comment this out if no Ackno
%The authors are grateful to the anonymous reviewers for finding
%a bug in the main result.

%%% AUTHOR:
%%% Bibliography goes here. Note that the arXiv cannot process bibtex
%%% or biber bibliographies.  Example of acceptable bibliograpy format:
\bibliographystyle{amsplain}

%% AUTHOR: You can generate such a bibliography from a .bib file by 
%% running pdflatex/bibtex/pdflatex/pdflatex and then pasting the .bbl file
%% between \begin{thebibliography} and \end{bibliography}

%%% AUTHOR: Include a short description of each author following the
%%% structure below. Use the same short tags used previously.  
%%% Use \imageat{} and \imagedot{} instead of "@" and "." in
%%% email addresses-this replaces the symbols with graphics to avoid 
%%% e-mail address harvesting from the .pdf file
\begin{dajauthors}
\begin{authorinfo}[jarek]
  Jaros\l aw B\l asiok\\
  Columbia University\\
  New York, United States of America\\
  jb4451\imageat{}columbia\imagedot{}edu
\end{authorinfo}
\begin{authorinfo}[patrick]
  Patrick Lopatto\\
  Brown University\\
  Providence, United States of America\\
  patrick\_lopatto\imageat{}brown\imagedot{}edu \\
  \url{https://lopat.to/}
\end{authorinfo}
\begin{authorinfo}[kyle]
  Kyle Luh\\
  University of Colorado Boulder\\
  Boulder, United States of America\\
  kyle\imagedot{}luh\imageat{}colorado\imagedot{}edu\\
  \url{https://sites.google.com/view/kluh/}
\end{authorinfo}
\begin{authorinfo}[Jake]
  Jake Marcinek\\
  Radix Trading\\
  Chicago, United States of America\\
  jakemarcinek\imageat{}gmail\imagedot{}com
\end{authorinfo}
\begin{authorinfo}[Shravas]
  Shravas Rao\\
  Portland State University\\
  Portland, United States of America\\
  shravas\imageat{}pdx\imagedot{}edu \\
  \url{https://web.cecs.pdx.edu/~shravas/}
\end{authorinfo}
\end{dajauthors}

\end{document}